\documentclass[compsoc,journal,onecolumn]{IEEEtran}
\usepackage{graphicx,amsthm,amsmath,amsfonts,latexsym,amssymb,multirow,booktabs}
\usepackage{color,array,verbatim,algpseudocode,bm}
\definecolor{menucolor}{rgb}{0.1,0.52,0.47}
\definecolor{urlcolor}{rgb}{0.85,0.37,0.01}
\definecolor{runcolor}{rgb}{0.46,0.44,0.701}
\definecolor{filecolor}{rgb}{0.2,0.5,0.01}
\definecolor{linkcolor}{rgb}{0.12,0.47,0.70}
\definecolor{citecolor}{rgb}{0.55,0.36,0.01}
\definecolor{anchorcolor}{rgb}{0.4,0.4,0.4}
\usepackage[colorlinks=true, urlcolor=urlcolor, linkcolor=linkcolor,citecolor=citecolor,filecolor=filecolor, anchorcolor=anchorcolor,menucolor=menucolor]{hyperref}
\usepackage{nameref}
\usepackage{zref-xr}
\usepackage{cite}

\zxrsetup{
  tozreflabel=false,
  toltxlabel=true,  
}
\usepackage[shortlabels]{enumitem}
\ifCLASSOPTIONcompsoc
  \usepackage[caption=false,font=normalsize,labelfont=sf,textfont=sf]{subfig}
\else
\usepackage[caption=false,font=footnotesize]{subfig}
\fi

\usepackage{array}
\usepackage{setspace}

\newtheorem{thm}{Theorem}
\newtheorem{prop}[thm]{Proposition}
\providecommand{\TR}{{\textrm{TR}}}
\providecommand{\TE}{{\textrm{TE}}}
\newcommand{\T}{{\textrm{T}}}
\newcommand{\mO}{\mathcal {O}}

\newcommand{\bd}{\boldsymbol{d}}
\newcommand{\bh}{\boldsymbol{h}}

\newcommand{\bLambda}{\boldsymbol{\Lambda}}

\newcommand{\bbeta}{\boldsymbol{\beta}}

\newcommand{\bPsi}{\boldsymbol{\Psi}}

\newcommand{\bnu}{\boldsymbol{\nu}}

\newcommand{\bOmega}{\boldsymbol{\Omega}}

\newcommand{\bc}{\boldsymbol{c}}
\newcommand{\bD}{\boldsymbol{D}}

\newcommand{\bH}{\boldsymbol{H}}

\newcommand{\bI}{\boldsymbol{I}}
\newcommand{\bJ}{\boldsymbol{J}}

\newcommand{\bS}{\boldsymbol{S}}

\newcommand{\bR}{\boldsymbol{R}}
\newcommand{\bW}{\boldsymbol{W}}

\newcommand{\bv}{\boldsymbol{v}}

\newcommand{\bx}{\boldsymbol{x}}

\newcommand{\bw}{\boldsymbol{w}}

\newcommand{\bzero}{\boldsymbol{0}}

\newcommand{\SE}{\textnormal{SE}}

\newcommand{\E}{{\mathbb{E}}}
\newcommand{\Var}{{\mathbb{V}{\rm ar}}}

\newcommand{\bit}{\begin{itemize}}
\newcommand{\eit}{\end{itemize}}
\newcommand{\ben}{\begin{enumerate}}
\newcommand{\een}{\end{enumerate}}
\newcommand{\beqn}{\begin{equation}}
\newcommand{\eeqn}{\end{equation}}
\newcommand{\bea}{\begin{eqnarray*}}
\newcommand{\eea}{\end{eqnarray*}}
\newcommand{\bpf}{\begin{proof}}
\newcommand{\epf}{\end{proof}\ms}
\newcommand{\ms}{\medskip}

\newcommand{\citep}{\cite}
\newcommand{\citet}{\cite}
\newcommand{\citeyear}{\cite}
\newcommand{\citeauthor}{\cite}

\newcommand{\A}{\mathbb{A}}
\newcommand{\I}{\mathbb{I}}
\renewcommand{\L}{\mathbb{L}}

\newcommand{\Tr}{\mathrm{Tr}~}
\def\shalf{\mbox{{\footnotesize$\frac{1}{2}$}}}

\newcommand{\W}{{\bm{W}}}

\newcommand{\1}{\bm{1}}
\newcommand{\D}{\bm{D}}

\usepackage{xcolor,colortbl}
\definecolor{SkyBlue}{rgb}{0.9,0.975,1}

\begin{document}
\title{Fast matrix-free methods for model-based personalized synthetic  MR imaging} 
\author{Subrata~Pal,~Somak~Dutta~and~Ranjan~Maitra
}



\IEEEcompsoctitleabstractindextext{
  \begin{abstract}
  Synthetic Magnetic Resonance (MR) imaging predicts images at new
  design parameter settings from a few observed MR scans.
  Model-based methods, that use
  both the physical and statistical properties underlying the MR
  signal and its acquisition, can predict images at any setting from as
  few as three scans, allowing it to be used in individualized patient- and
  anatomy-specific contexts. However, the estimation problem in
  model-based synthetic MR imaging is ill-posed and so regularization,
  in the form of  correlated Gaussian 
  Markov Random Fields, is imposed on the voxel-wise spin-lattice
  relaxation time, spin-spin relaxation time and the proton density
  underlying the MR image. We develop theoretically sound but
  computationally practical matrix-free estimation methods for synthetic
  MR imaging. Our evaluations demonstrate superior performance of our
  methods in currently-used clinical settings when compared to existing model-based
  and deep learning methods. Moreover, unlike deep learning approaches,
  our fast methodology can   synthesize needed images during patient visits,
  with  good estimation and prediction accuracy and consistency.  An added
  strength of our model-based approach, also developed and illustrated
  here, is the accurate   estimation of standard errors  of regional
  contrasts in the synthesized images. A R package {\tt symr} implements
  our methodology.
\end{abstract}



}
\maketitle

\section{Introduction} 
\label{sec:intro}
Magnetic Resonance~(MR) Imaging (MRI)~\citep{katims82,
  mansfieldandmorris82,hinshawandlent83} is a noninvasive 
radiologic tool for understanding tissue structure and 
function.  In MRI, each tissue type has a distinctive  spin-lattice 
or longitudinal  relaxation time ($\T_1$), spin-spin or transverse
relaxation time ($\T_2$), and proton density ($\rho$). 
These quantities are observed only indirectly through  
noise-contaminated measurements of their transformations that
modulate their effect through user-controlled design parameters such as
echo time (TE), repetition time (TR) or flip angle ($\alpha$). 
The exact way in which $\rho$, $\T_1$ and $\T_2$ combine with the
design parameters to form an MR image depends on the imaging sequence
used~\citep{kuperman00,bernsteinetal04,brownetal14}. For example, in
spin-echo MRI (where $\alpha=\!90^\circ\!$), 
the noiseless  MR intensity $\nu_{ij}$
at the $i$th voxel
and $j$th design parameter setting $(\mbox{TE}_j, \mbox{TR}_j)$ is approximately
specified by the Bloch equation
\begin{equation}
	\nu_{ij} = \rho_i \exp\left(-\frac{\TE_j}{\T_{2i}}\right)\left\{1 -
	\exp\left(-\frac{\TR_j}{\T_{1i}}\right)\right\}.
	\label{eq:bloch}
\end{equation}
In reality, the observed MR signal is noise-contaminated and 
complex-valued, but 
its magnitude ($r_{ij}$, at the $i$th
voxel and $j$th design parameter setting) is commonly stored, and is well-modeled~ by a Rice distribution~\citep{henkelman85, 
  sijbers98, sijbersetal98,rice44, rice45} with density
\begin{equation}
	\varrho(r_{ij}; \sigma_j, \nu_{ij}) = \frac{r_{ij}}{\sigma_j^2} \exp\left(-\frac{r_{ij}^2 + \nu_{ij}^2}{2\sigma_j^2}\right)
	\I_0\left(\frac{r_{ij}\nu_{ij}}{\sigma_j^2}\right),
	\label{eq:rice}
\end{equation}
for $r_{ij}>0$. Here $\sigma>0$ is the scale parameter of the Rice
density and $\I_k(\cdot)$ is the modified Bessel function of the first
kind of order $k$. The distribution of $r_{ij}$ follows from the
approximate distribution~\citep{wangandlei94} of the complex-valued MR signal as
bivariate Gaussian with mean $(\nu_{ij}\cos\eta,\nu_{ij}\sin\eta)$ and homogeneous diagonal dispersion matrix $\sigma^2\bI_2$, for some $\eta$. For identifiability and convenience,  we use $\eta=0$~\citep{maitraandfaden09,maitraandriddles10,maitra13}.

For a spin-echo MR imaging sequence~\citep{hahn50,hennigetal86},
different $(\TE,\TR)$ values modulate $\T_1$, $\T_2$ and $\rho$
and can be used to enhance the clinical distinction between tissues. 
However, the optimal $(\TE, \TR)$-settings are patient- and/or
anatomy-specific and not always known in advance.
Further, patient discomfort or technological
challenges can preclude the acquisition of images at some
settings~\citep{deonietal05,maitraandriddles10}. These shortcomings
were sought to be overcome using {\em synthetic MRI}
(syMRI)~\citep{bobmanetal85, 
  bobmanetal86, feinbergetal85, ortendahletal84b,gonccalvesetal18},
where the underlying $\rho,\T_1,$ and $\T_2$ values at each voxel are
estimated using a set of {\em training} images, and then inserted along 
with unobserved $(\TE,\TR,\alpha)$-values into~\eqref{eq:bloch} to
obtain corresponding synthetic images. The technique is an 
appealing alternative 
in situations requiring images at  multiple settings and/or long
acquisition times, such as in pediatric imaging of the developing
brain or in non-cooperating subjects~\citep{bettsetal16, mcallisteretal17, 
  andicaetal19,hagiwaraetal17b}.
However, despite its vast potential, syMRI's clinical adaptation has
been stymied by the inherent ill-posedness of the inversion of the Bloch
equation, which leads to unstable
solutions~\citep{maitraandriddles10}, or by the inordinate times taken
by palliative measures such as regularization using Markov Random
Field (MRF) priors~\citep{gladandsebastiani95,maitraandbesag98}. Deep
learning (DL) methods were developed \citep{hagiwaraetal19}, but they
require substantially large training images that are impractical in a
clinical single-subject personalized setting. Recently,
\citet{paletal22} developed a DIPsynMRI framework based on a ``deep
image prior'' \citep{ulyanovetal18} that works well in personalized
settings with only three training images. However, besides being
computationally slow, these DL methods do not easily provide any measure of prediction uncertainty.
An alternative model-based approach proposed by~\citet{maitraandriddles10}, while
practical, used the one-step late Expectation-Maximization (EM),
or OSL-EM algorithm~\citet{green90} that is not guaranteed to
converge, and, as shown in this paper, can provide sub-optimal estimates. 


In this paper, we  
instead develop (Section~\ref{sec:meth}) a 
scalable  alternating expectation conditional maximization (AECM)
algorithm~\citep{mengandvandyk97} for synthetic MRI that is guaranteed
to converge and essentially matches the speed of non-regularized
methods such as LS. The key ingredients of our AECM algorithm are
a checkerboard coloring scheme~\citep{besag74,besagetal95} that allows
fast parallel block maximization over the voxel-wise
parameters and  analytically profiles out some of the Gaussian MRF (GMRF) parameters,
reducing the optimization over the remaining parameters to be of at most two
variables. An additional appealing aspect of our model-based approach
is the conceptual ease of obtaining standard errors (SEs) of
our predictions, for which we also provide methods for their fast
implementation. 
Section~\ref{sec:sim} evaluates properties of our estimates, and demonstrates our methodology on a normal subject
in a clinical setting. We also investigate the selection of test sets
for optimizing predictive performance and demonstrate consistency of
our synthetic MRI predictions. A separate set of simulation
experiments evaluates estimation and prediction accuracy and
consistency for different noise levels.  
We conclude with some discussion.
An online supplement with sections,
figures and tables prefixed by `S', is available. 
Our methods 
are implemented in a 
R~\citep{R} package 
\textsf{symr} (pronounced $sim\!\cdot\!mer$) 
that uses openMP-based parallelization for computational practicality. 



\section{Methodology}\label{sec:meth}
\subsection{A penalized likelihood model for syMRI}
Following \citet{maitraandriddles10}, the observed
image intensities $\bR = \{r_{ij}\}$ are assumed to be  
independent Rice-distributed random variables with
density~\eqref{eq:rice}. We  get the loglikelihood 
\begin{equation}
	\ell(\{\T_{1i},\T_{2i},\rho_i\}_{i=1}^{n};\bR) =
	\sum_{i=1}^n \sum_{j=1}^m \log\varrho(r_{ij}; \nu_{ij}, \sigma_j),
	\label{llhd}
\end{equation}
where the voxels $i=1,2,\ldots,n$ are embedded in a regular 
3D array of $n_x\!\times\! n_y\!\times\! n_z$ voxels. The noise parameter $\sigma_j^2$ in
each training image is estimated using~\citet{maitra13}. As in
\citet{maitraandriddles10}, we spatially regularize
$(\rho,\T_{1},\T_{2})$ by the voxel-wise transformations $W_{i1} = \rho_i$,
$W_{i2} = \exp\left(-1/\T_{1i}\right)$, and $W_{i3} =
\exp\left(-1/\T_{2i}  \right)$ and penalizing
~\eqref{llhd} to obtain the optimization problem 
\begin{equation}\label{eq:penalizedproblem}
\max_{\W} \left\{\ell(\W;R) - \shalf \Tr\bPsi^{-1}\W'\bLambda\W\right\}. 
\end{equation}
In~\eqref{eq:penalizedproblem}, $\W = (W_{ik})$ is a $n\!\times\!3$ matrix, $\bPsi$ is a $3\!\times\!3$ positive definite matrix and $\bLambda = \beta_z \bJ_{n_z}\! \otimes\! \bI_{n_y}\! \otimes\! \bI_{n_x}\! +\! \beta_y \bI_{n_z}\! \otimes\! \bJ_{n_y}\! \otimes\! \bI_{n_x}\! +\! \beta_x \bI_{n_z}\! \otimes\! \bI_{n_y}\! \otimes\! \bJ_{n_x},$
where $\beta_x$, $\beta_y$, and $\beta_z$ are nonnegative parameters,
$\bI_{k}$ is the identity matrix of order $k,$ and $\bJ_k$ is the $k\!\times\!k$ tri-diagonal matrix with all nonzero off-diagonals as $-1$, the first
and the last diagonal entries both as $1$ and the other 
diagonal entries all as $2$. For all practical purpose, $n\!=\!n_x n_y
n_z$ is large, and $\bLambda$ is an enormous but sparse matrix,
with each row having at most seven non-zero elements, corresponding to the
two neighbors in each direction and the voxel itself. 
The penalty in~\eqref{eq:penalizedproblem} can be written as the
kernel of  the (improper) matrix normal density 
\begin{equation}
	f(\W; \bPsi, \bbeta) =
	\frac{|\bLambda|_+^{3/2}\exp\left\{-\shalf\Tr\left(\bPsi^{-1}\W'\bLambda \W\right)\right\}}{(2\pi)^{3n/2}|\bPsi|^{n/2}},
	\label{matrixnormal}
\end{equation}
where $|\bLambda|_+$ is the product of the positive eigenvalues of
$\bLambda$. The Kronecker structure in the components of $\bLambda$
means that $\bPsi$ and $\bbeta=(\beta_x,\beta_y,\beta_z)$ need constraints for
identifiability. Though there are many options, we
impose the constraint that $2\beta_x+2\beta_y+2\beta_z\! =\! 1,$ from the alternative specification of \eqref{matrixnormal} as the
(improper) density of a first order intrinsic multi-layer GMRF~\citep{besag74,besagandkooperberg95,duttaandmondal15}. 
If $\W_i'$ is the $i$th row (corresponding to the $i$th interior
voxel) of $\W,$ then the conditional distribution of $\W_i$ given
$\bW_{-i}$ (all other rows of $\W$)  is trivariate normal with  conditional mean and variance
\begin{equation}\label{eq:cond_mean}
	\E(\W_i | \bW_{-i}) =  -\sum_{q \heartsuit i} \Lambda_{iq} \W_q,\qquad\qquad \Var(\W_i|\bW_{-i}) = \bPsi,
\end{equation}
where $q\heartsuit i$ if and only if voxel $q$ is a neighbor of voxel
$i$, and $\Lambda_{iq}$ is the $(i,q)$th element of $\bLambda.$
Similar expressions exist for boundary voxels. Further, 
$\Lambda_{iq} = -\beta_x,$  $-\beta_y$ or $-\beta_z$ if $q\heartsuit
i$ in the $x$-, $y$- or $z$-direction, so~\eqref{eq:cond_mean}
characterizes $\E(\bW_i|\bW_{-i})$ as the weighted average of its
neighboring $\W_q$ values with weights $\beta_x,\beta_y$ or
$\beta_z$. The conditional covariance is constant. 
\subsection{Matrix-free AECM for parameter estimation}
\subsubsection{Background}\label{sec:bg}
Analytical or numerical optimization in~\eqref{eq:penalizedproblem}
is well-nigh impractical because of the intractability
introduced by the penalty term. A EM algorithm can conceptually be
developed from the generative model of the Rice distribution of the
individual $r_{ij}$s.  
Specifically, with $r_{ij}$ as the observed magnitude, we let
$\theta_{ij}$ be the (missing) phase angle, then
$(r_{ij}\cos\theta_{ij}, r_{ij}\sin\theta_{ij})$ is bivariate normally
distributed with mean $(\nu_{ij}, 0)$ and variance-covariance matrix
$\sigma_j^2\bI_2$. Then, ignoring terms not involving $\nu_{ij}$, the
complete log-likelihood is $ \sum_{i=1}^n \sum_{j=1}^m \sigma_j^{-2}
(r_{ij}\nu_{ij}\cos\theta_{ij} - \nu_{ij}^2/2)$. The algorithm
iteratively  computes, in the expectation step (E-step), the
expectation of the complete loglikelihood given the observations, and 
evaluated at the current parameter values, while the maximization step
(M-step) maximizes
the result. We now list the two steps.
\paragraph{E-step}  The E-step requires the conditional expectation of
each $\cos\theta_{ij}$ given $r_{ij}$ at the current values 
($\W^{(t)},\bPsi^{(t)},\bbeta^{(t)}$) of ($\W,\bPsi, \bbeta$).
We compute and define
$$Z_{ij}^{(t)} := \E[\cos\theta_{ij}|r_{ij}; \nu_{ij}^{(t)}] =  \A_1\left(r_{ij}\nu_{ij}^{(t)}/\sigma_j^2\right),$$
where $\A_1(x) := \I_1(x)/\I_0(x).$
Here, $Z_{ij}^{(t)}$ is free of $\bbeta^{(t)}$ or $\bPsi^{(t)}$, but rather, depends only 
on $\nu_{ij}^{(t)}$ which only involves the $i$th row of $\W^{(t)}$. So $Z_{ij}^{(t)}$s  
can be computed in parallel without any data racing by allocating an instance of $\sigma_j^2$ to every processor thread.

\paragraph{M-Step} At the $(t\!+\!1)$th iteration, updates $(\bW^{(t+1)},
\bPsi^{(t+1)}, \bbeta^{(t+1)})$ are obtained by maximizing 
\begin{equation}
Q^*(\bW, \bPsi, \bbeta;\bW^{(t)}) = Q(\bW;\bW^{(t)}) + \log f(\bW;
\bPsi, \bLambda)
\end{equation}
with respect to (w.r.t.) $(\bW,\bPsi,\bbeta),$ 
with
\begin{equation}
	Q\left(\W;\W^{(t)}\right) = \sum_{i=1}^n \sum_{j=1}^m \sigma_j^{-2}\left( -\shalf\nu_{ij}^2 + r_{ij}\nu_{ij}Z_{ij}^{(t)}\right).
	\label{eq:Q_fn}
      \end{equation}
      In \eqref{eq:Q_fn}, $\nu_{ij}$ is again a function only of the $i$th
row of $\W.$ However, optimization of $Q^*$ is still challenging 
so \citet{maitraandriddles10} employed an OSL-EM
algorithm \citep{green90}. Specifically, they solved $\nabla
Q(\W|\W^{(t)}) + \nabla \log f(\W^{(t)};\bPsi,\bbeta) = 0$ in
$(\bW,\bPsi\bbeta)$. This equation can be decomposed  into those involving
only the individual rows of $\W$, greatly simplifying computation. 
(The equations for $\bPsi$ and
$\bbeta$ are still joint.) However, unlike the EM algorithm, 
OSL-EM is neither guaranteed to converge~\citep{maitraandriddles10}, nor increase the penalized
log-likelihood monotonically at each iteration (see Figure~\ref{fig:likeliplot} for an illustration where OSL-EM actually
reduces the penalized log-likelihood providing a sub-optimal
estimate). So, we describe and implement an AECM algorithm that is guaranteed to converge monotonically to a local maximum.
\subsubsection{Development and implementation}
An AECM algorithm~\citep{mengandvandyk97} partitions the parameter
space and splits the M-step into conditional maximization (CM)
steps, one for each element in the above partition, and alternates the
E-step with each CM step. In our framework, we exploit the
GMRF structure underlying $\bW$ to  color the
voxels as black or white using a checkerboard
pattern~\citep{besag74} so that no two neighboring voxels have the same color. Then
our parameter space is partitioned into three sets, with two partitions
for the $\bW$s at the black and white voxels
and the third partition containing
$(\bPsi,\bbeta)$. Then each partition is updated in its CM-step while
keeping the variables in the other partitions fixed at their current values.

Our bifurcation of $\bW$ is driven by the fact that
the CM-step for $\bW$  at each black (or white) voxel
involves fixed values from only the neighboring white (correspondingly
black) voxels. So the optimization problems over the black
(white) voxels are {\em embarrassingly} parallel and are
efficiently implemented with no data races in computer memory.
Our race-negligible parallel optimization is explicitly given by 
\begin{prop}\label{prop:optimization_W_i}
Maximizing
$Q^*(\bW,\bPsi,\bbeta;\bW^{(t)})$ w.r.t. $\W_i$ for fixed
$(\bPsi,\bbeta)$ and 
$\bW_{-i}$ is equivalent to maximizing
\begin{equation}\label{eq:Q_star_i}
  Q^*_i(\W_i) = \sum_{j=1}^m \frac1{\sigma_j^{2}}\left( -\nu_{ij}^2/2 + r_{ij}\nu_{ij}Z_{ij}^{(t)}\right) - \sum_{q \heartsuit i} \Lambda_{i,q} \W_q' \bPsi^{-1}\W_i -\shalf \Lambda_{i,i} \W_i' \bPsi^{-1} \W_i,
              \end{equation}
	with gradient vector w.r.t. $\W_{i}$ given by
	\begin{equation}\label{eq:grad_Q_star_i}
            \nabla Q^*_i(\W_i)
            =  \sum\limits_{j=1}^m \frac1{\sigma_j^{2}}\left(
		- \nu_{ij} +
		r_{ij}Z_{ij}^{(t)}\right)\nabla \nu_{i,j}
-\sum\limits_{q\heartsuit i} \Lambda_{i,q} \bPsi^{-1}W_q  - \Lambda_{i,i} \bPsi^{-1}\W_i.
\end{equation}
Here, $\nabla \nu_{ij}$ is the gradient vector of $\nu_{ij}$ w.r.t. $\W_i$.
\end{prop} 
\begin{proof}
See Section~\ref{proof1}.
\end{proof} 

Proposition \ref{prop:optimization_W_i} yields a small 3D optimization problem that we efficiently solve using the 
quasi-Newton box-constrained L-BFGS-B algorithm with limited memory
calculations for the Hessian matrix \citep{byrdetal95, zhu97}. The box
constraints are obtained from the ranges of $\bW$. 
Further, the part of \eqref{eq:MRF_part} that is inside the
parenthesis, is needed to be computed only once per iteration per voxel using sparse
matrix computations and the last term can be computed many times
inside the L-BFGS-B algorithm as the algorithm explores the
space. Also, to have more stable ratios of Bessel functions in the
expression for $Z^{(t)}_{ij}$, we have used exponentially scaled
Bessel functions~\citep{abramowitzandstegun72}, as in \citet{maitra13}, and resorted to a Taylor series expansion as the argument tends to $\infty$. 


Next, for fixed $\W$, we describe how to update the estimates of
$(\bPsi,\bbeta).$ The following proposition shows that given $\bbeta,$
$\bPsi$ can be analytically profiled out giving a simple objective function only in $\bbeta$.
\begin{prop}\label{prop:profileout}
 Fix $\bW=\bW^{(t+1)}.$ For fixed $\bbeta$, the function $\bPsi \rightarrow Q^*(\W,\bPsi,  \bbeta; \bW)$  is maximized at $\displaystyle \hat\bPsi =
  n^{-1}\W'\bLambda\W$ (with $\bLambda\equiv\bLambda(\bbeta)$).
  The resulting profile function for $\bbeta$ is
	\begin{equation}\label{eqn:profilelikelihood}
		Q^*_p(\bbeta) = c + \frac{3}{2}\log |\bLambda|_+ - \frac{n}{2} \log |\W'\bLambda \W|
	\end{equation}
	where $c$ is a constant that depends only on $n$ but not on $\bbeta$. 
\end{prop}
\begin{proof}
See Section~\ref{proof2}.
\end{proof}
\eqref{eqn:profilelikelihood} and its derivative require calculation
of $|\bLambda|_+$ and the determinant of a
$3\!\times\!3$ matrix $\W'\bLambda \W.$ 
but these are easily obtained by noting that the eigenvalues of
$\bJ_k$~\citep{duttaandmondal15} are
\begin{equation}
	\lambda_i = 2\left[1-\cos\left\{\pi(i-1)/k\right\}\right] \mbox{ for } i=1,2 \ldots, k.
\end{equation}
Let $\bD_k$ be the vector of eigenvalues of $\bJ_k\ (k = n_x,
n_y, n_z)$. Then $|\bLambda|_+$ is a product of the elements of the
vector $\beta_x 
\1_{n_z} \!\otimes\! \1_{n_y}\! \otimes\! \D_{n_x}\! +\! \beta_y \1_{n_z}\! \otimes
\D_{n_y}\! \otimes\! \1_{n_x}\! +\! \beta_z \D_{n_z}\! \otimes\! \1_{n_y}\! \otimes\!\1_{n_x}$.
Our matrix-free treatment of $(\bPsi,\bbeta)$ is a faster, more elegant and
easier-implemented contrast to~\citet{maitraandriddles10}
where $\bPsi$ is further parametrized via a Cholesky factorization and
updated along  with $\bbeta$ using another L-BFGS-B algorithm.

Like OSL-EM, we initialize $\bW$ in our AECM algorithm with the LS
estimates, and obtain initial values of $\bPsi$ and $\bbeta$, as per
Proposition~\ref{prop:profileout}. Unlike OSL-EM, our AECM algorithm
is guaranteed to converge~\citep{mengandvandyk97}, as also shown in
Figure~\ref{fig:likeliplot}.

\subsubsection{Generating synthetic images}
The estimated $(\rho,\T_1,\T_2)$, obtained from the terminating
$\bW$s, are inserted along with desired design parameter
settings into~\eqref{eq:bloch} to synthesize spin-echo images. Similar methods
can be employed for other MRI pulse sequences.

\subsection{Matrix-free standard error computations}
\label{variance}
An attractive feature of our model-based approach and the use of penalized
likelihood estimators is the provision for ``on-the-fly'' calculation of
SEs. We now provide the theoretical
development and fast implementation for calculating SEs of contrasts of
regions of interest (ROIs) in synthetic images $\widehat\bnu$.

Following \citet{segaletal94} or~\citet{leeandpawitan14}, the estimated observed penalized information matrix for $\W$ is given by
$\widehat{\bOmega} = \widehat\bH + \widehat{\bLambda}\otimes \widehat{\bPsi}^{-1},$
where $\widehat{\bH}$ is the negative of the Hessian of the observed
log-likelihood \eqref{llhd} with respect to $\W$ evaluated at $\widehat{\W}$
(see Section~\ref{sec:supp-info-mat}), $\widehat{\bLambda}$ is the
$\bLambda$ matrix evaluated at the terminating $\widehat{\bbeta}$, and
$\widehat{\bPsi}$ is also obtained from the terminated AECM algorithm. Thus,
$\widehat\bH$ is a $3n\!\times\! 3n$ block-diagonal matrix (with entries
corresponding to the background voxels set to zero). 
The approximate dispersion matrix of a synthetic image $\hat\bnu$ is calculated
from $\hat\bOmega$ by the multivariate delta method to be $\bS= \nabla\widehat{\bnu} \widehat{\bOmega}^{-1}\nabla \widehat{\bnu}'$
where $\nabla\widehat{\bnu}$ is the $n\times 3n$ Jacobian matrix of
$\bnu$ with respect to $\W$ evaluated at $\widehat{\W}$.  The computation of
$\bS$ is challenging but can be parallelized. Moreover, in 
practice, we are mostly interested in calculating means and SEs of
ROIs drawn on synthetic images. The SE of a linear combination
$\bc'\widehat\bnu$ of the voxel-wise values of a synthetic image is
$\SE(\bc'\widehat{\bnu}) = ({\bc'\bS\bc})^{1/2} =
\{(\bc'\nabla\widehat{\bnu})
\widehat{\bOmega}^{-1}(\nabla\widehat{\bnu}'\bc)\}^{1/2}.$

The computation of
$\SE(\mathbf{c}'\widehat{\bnu})$ still presents substantial numerical
challenges.
In particular, it requires us to solve an equation of the form
$\widehat{\bOmega}\bx = \bd$ where $\bd = (\nabla
\widehat{\bnu})'\bc$. Although $\widehat{\bOmega}$ is sparse, it is
typically so large that standard methods based on sparse matrix
factorization (e.g. sparse Cholesky decomposition) run out of memory
mainly because they suffer from fill-ins. In our example, even the
heuristic methods \citep[Ch. 8]{davis2006direct} for reducing fill-ins
failed to provide a factorization in machines with as much as 384GB RAM. So
we use the matrix-free Lanczos algorithm~\citep{duttaandmondal15} to
solve the system. 
Specifically, we initialize as $\theta_1 =
\|\bd\|$, $\bv_1 = \bd/\theta_1,$ $\bw_1 = \widehat{\bOmega}\bv_1,$
$\kappa_1 = {\bv}_1'{\bw}_1,$ $\gamma_1 = \sqrt{\kappa_1},$ $\phi_1 =
\theta_1/\gamma_1,$ ${\bh}_1 = {\bv}_1/\gamma_1$ and set the initial
solution: ${\bx}_1 = \phi_1{\bh}_1$. Next, for $i=2,3,\ldots,$ we set
 \begin{itemize}
  \item  $\theta_i = \| {\bw}_{i-1} - \kappa_{i-1}{\bv}_{i-1} \|.$
  \item  ${\bv}_i = \left({\bw}_{i-1} - \kappa_{i-1}{\bv}_{i-1}\right)/\theta_i$.
  \item  ${\bw}_i = \widehat{\bOmega}{\bv}_{i} - \theta_i {\bv}_{i-1}.$
  \item $\kappa_{i} = {\bv}_{i}'{\bw}_i,$ $\delta_i = \theta_i/\gamma_{i-1},$ and $\gamma_i = \sqrt{\kappa_i - \delta_i^2}.$ 
  \item  $\phi_i = -\delta_i\phi_{i-1}/\gamma_i$ and ${\bh}_i = \left(
      {\bv}_i - \delta_i{\bh}_{i-1} \right)/\gamma_i.$
  \end{itemize}
  and update the solution: ${\bx}_i = {\bx}_{i-1} + \phi_{i}
  {\bh}_i$ till convergence, that is declared if the relative error
  $\|{\widehat{\bOmega}\bx}_i - \bd\|/\theta_1$ is below some pre-specified tolerance.

By carefully recycling the memory space for ${\bv}_i$s and ${\bw}_i$s,
the Lanczos algorithm requires only $\mO(n)$ memory in addition to the storage
for $\widehat{\bOmega}.$ Also, $\widehat{\bOmega}$ is used only
through the matrix-vector product $\widehat{\bOmega}{\bv}_{i}$ which can be
computed fast with $O(n)$ computational complexity. Finally, while we may use an
incomplete Cholesky preconditioner to accelerate convergence of the Lanczos
iterations~\citep{duttaandmondal15}, our examples showed that a simple
diagonal preconditioner achieves at least the same  computing speed,
and is often faster.

We conclude our discussion of SE estimation of regional contrasts in
synthesized MR images by noting that its feasibility stems from our
use of a direct model-based approach. Such an approach is not possible
with the deep learning methods of \citet{paletal22} because of the
intractability of the estimation. Moreover, a parametric bootstrap
approach to obtain SEs is not computationally feasible given the time
taken (over a day) to obtain estimates from one training set. In
contrast, our AECM estimates take not even minutes to obtain, which means that
they, along with regional contrast SEs, can be obtained and used
during a patient visit. 

This section has  developed a computationally practical AECM
algorithm for generating synthetic MR images from as few as three
training images, and a method for fast computations
of the SEs of ROI contrasts in such images.


\section{Performance evaluations} 
\label{sec:sim}
We illustrate our AECM methodology and evaluate its performance. Our
studies are on the only known 3D dataset  
collected on the same subject at multiple (TE,TR) test set values (in addition
to the training set) that permits
evaluation of predictions~\citep{maitraandriddles10,paletal22}, and
for greater understanding, on  
simulated datasets from the Brainweb interface~\citep{cocoscoetal97}.
(See Table~\ref{table:3D_best_10_whole_v2} for the design parameter
settings at which the dataset: the images are indexed accordingly.)

From each training set of three images, we estimated
$\rho,\T_1,\T_2$ using 
OSL-EM (OSL), DIPSynMRI (DL), and AECM, and then
inserted them  into~\eqref{eq:bloch} to synthesize
images at the test set values. Additionally, at the suggestion of a
reviewer, we also evaluated performance of our methodology when
\eqref{eq:rice} is approximated by a Gaussian distribution 
 with mean $\nu_{ij}$ and 
standard deviation $\sigma_j.$ The penalty remains unchanged, and is
 \eqref{matrixnormal}. In this scenario, the penalized maximum likelihood
 estimator is the same as the Penalized Least Squares (PLS) estimator
 (and referred to in this paper as such).
 Further, the estimation algorithm is then
 simply a block-coordinate descent method (note that we minimize the
 PLS), and the obtained estimates 
are again used along with 
~\eqref{eq:bloch} to synthesize
images at the desired test set values.

Whether obtained by OLS, DL, AECM or PLS, the predicted synthetic images were
compared to the test set images using the scaled root mean squared
prediction error (RMSPE) for the $j$th image volume, given by
$s_j^{-1}\sqrt{\sum_{i=1}^n (r_{ij} - \widehat\nu_{ij})^2/n},$ where
$s_j$ is the standard deviation of the foreground voxel values in 
$j$th acquired image. Note that our evaluations calculated
$\widehat\nu_{ij}$ in the (scaled) RMSPE calculations in
$\widehat\nu_{ij}$ in two ways. First, we used 
$\widehat\nu_{ij}$ without accounting for the Rice distribution
of the voxel-wise predictions. In the second case, we set $\widehat\nu_{ij}^* =
\hat\sigma_j\sqrt{\pi/2} 
\L_{1/2}\left(-{\hat{\nu}_{ij}^2/\hat{\sigma}_j^2}\right)$, or the mean
of the Rice distribution with parameters $\widehat\nu_{ij}$ and
$\hat\sigma_j$, with $ \L_{1/2}(x) =
e^{x/2}\left[\left(1-x\right)\I_{0}\left(-{x/2}\right)-x\I_{1}\left(-{x/2}\right)\right]$
denoting the Laguerre polynomial of order $1/2$. However, the RMSPEs
in each case were essentially the same, so we report our measures
using the predicted $\widehat\nu_{ij}$. Also, our computations were
done in parallel on a 2.3GHz Intel(R) Xeon(R) Gold 6140 processor,
capable of handling up to 36 threads.  
\subsection{Real Dataset}
\label{sec:ill}
The dataset was acquired~\citep{maitraandriddles10} on a
1.5T Signa scanner using a spin-echo imaging sequence at a resolution of
$1.15 \text{mm}\times 1.15 \text{mm}\! \times\! 7.25 \text{mm}$ in a
field of view set to be $ 294 \text{mm}\!\times\! 294 \text{mm}\!\times\!145 \text{mm}$. For each 3D image volume,  $n_x\!=n_y\!=\!256$ and
$n_z\!=\!20$.
\begin{figure}[h]
	  \centering
          \includegraphics[width=\linewidth]{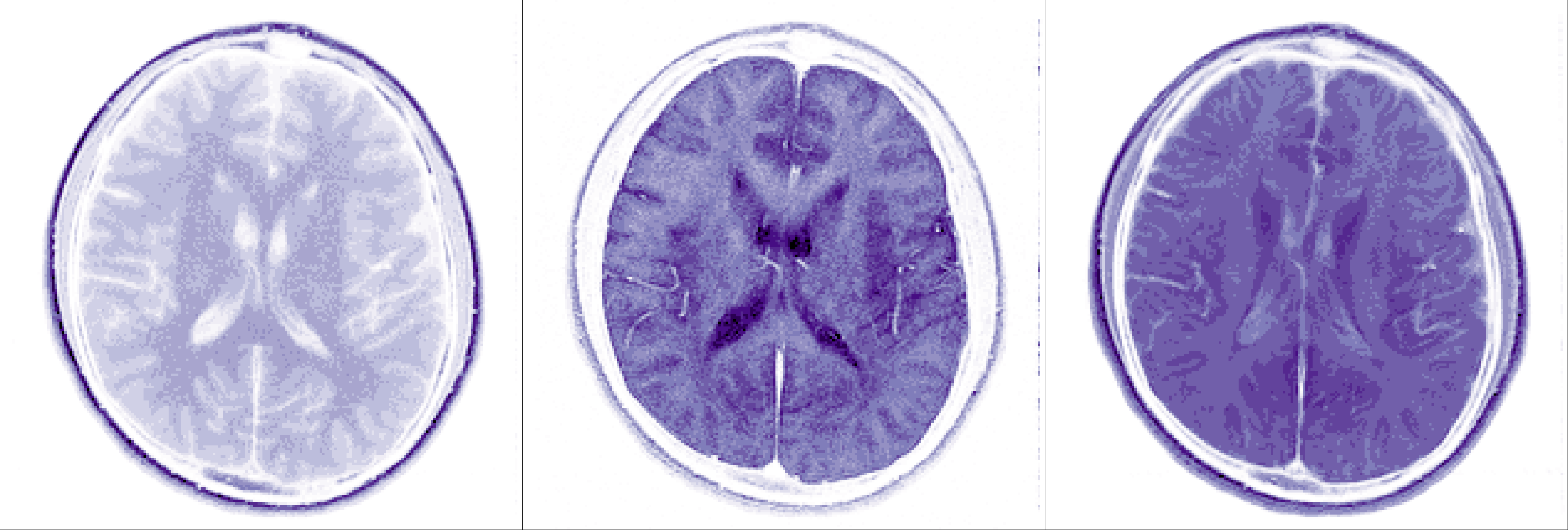}
	  \caption{Mid-axial view of the training set  images with (from left to right) $\TE/\TR =0.01/0.6, 0.08/2, 0.01/3$ (seconds).}  
	  \label{fig:3dplots_train}
\end{figure}
Figure~\ref{fig:3dplots_train} displays the 10th slice of the training
set images. Figure~\ref{fig:likeliplot} shows the value of 
~\eqref{eq:penalizedproblem} as the OSL-EM and AECM algorithms 
proceed from LS initialization. In this example, OSL-EM is seen to
have an uncertain trajectory, and on the whole, goes down after the first iteration. On
the other hand, AECM correctly increases and terminates in a few
iterations. Our implementation of AECM ensures that this accuracy does 
not come at a price, with estimation and synthetic image generation typically
done in about a minute. Further, the AECM shows greater stability with
respect to different initial values than OSL-EM. 

Figure \ref{fig:RMSPE} displays the performance in terms of the RMSPE,
of the competing methods on the eight acquired $(\TE, \TR)$ settings.
\begin{figure}[h]
  \centering
  \includegraphics[width=\linewidth]{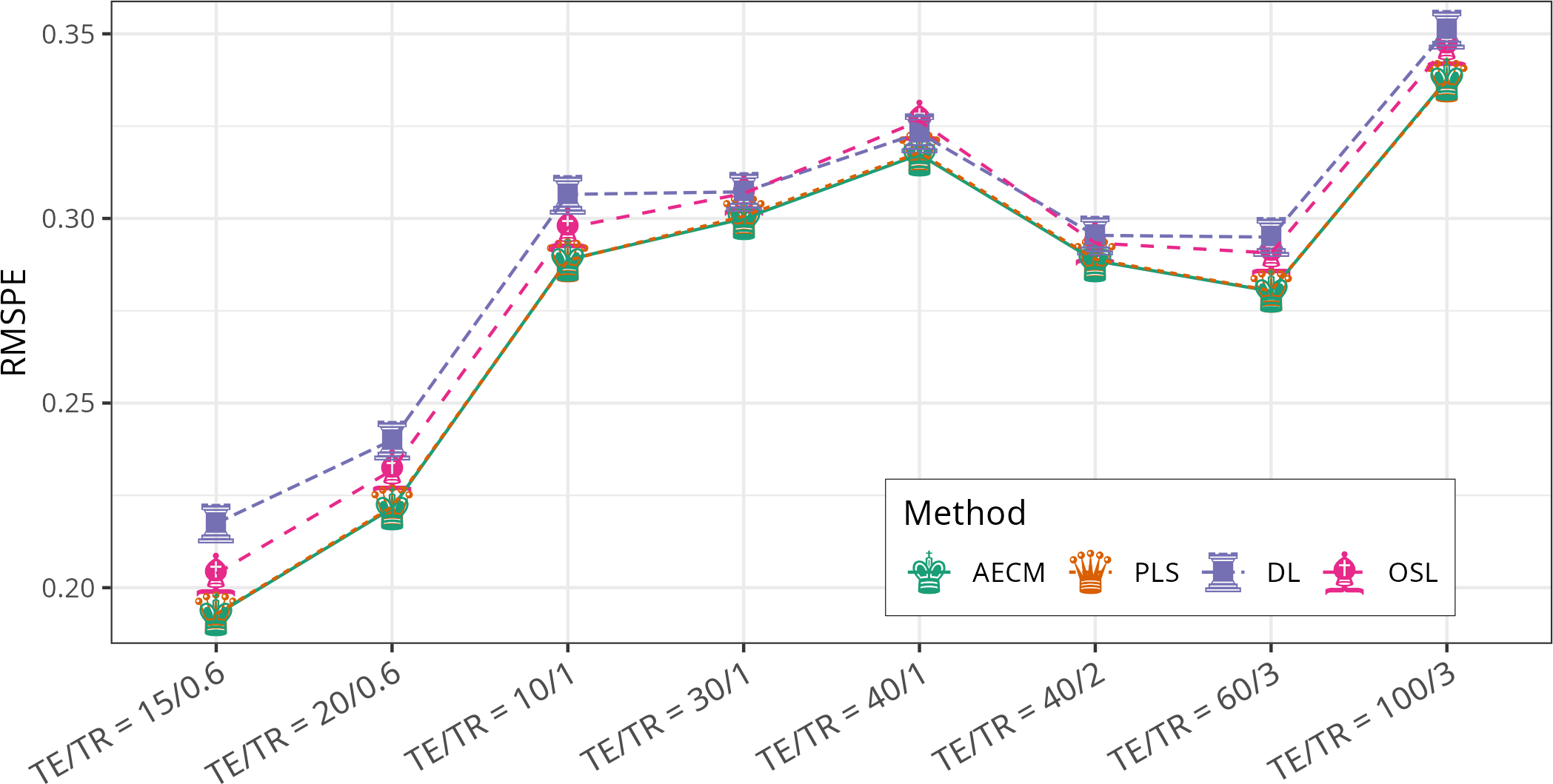}
  \caption{RMPSEs on the eight test images for the different methods.}
  \label{fig:RMSPE}
\end{figure}

We see that AECM is the best performer at all settings, closely followed by PLS.
\begin{figure}[h]  
  \centering
  \vspace{-0.1in}
  \mbox{
	   \subfloat{
		   \includegraphics[width=0.3\linewidth, page=1]{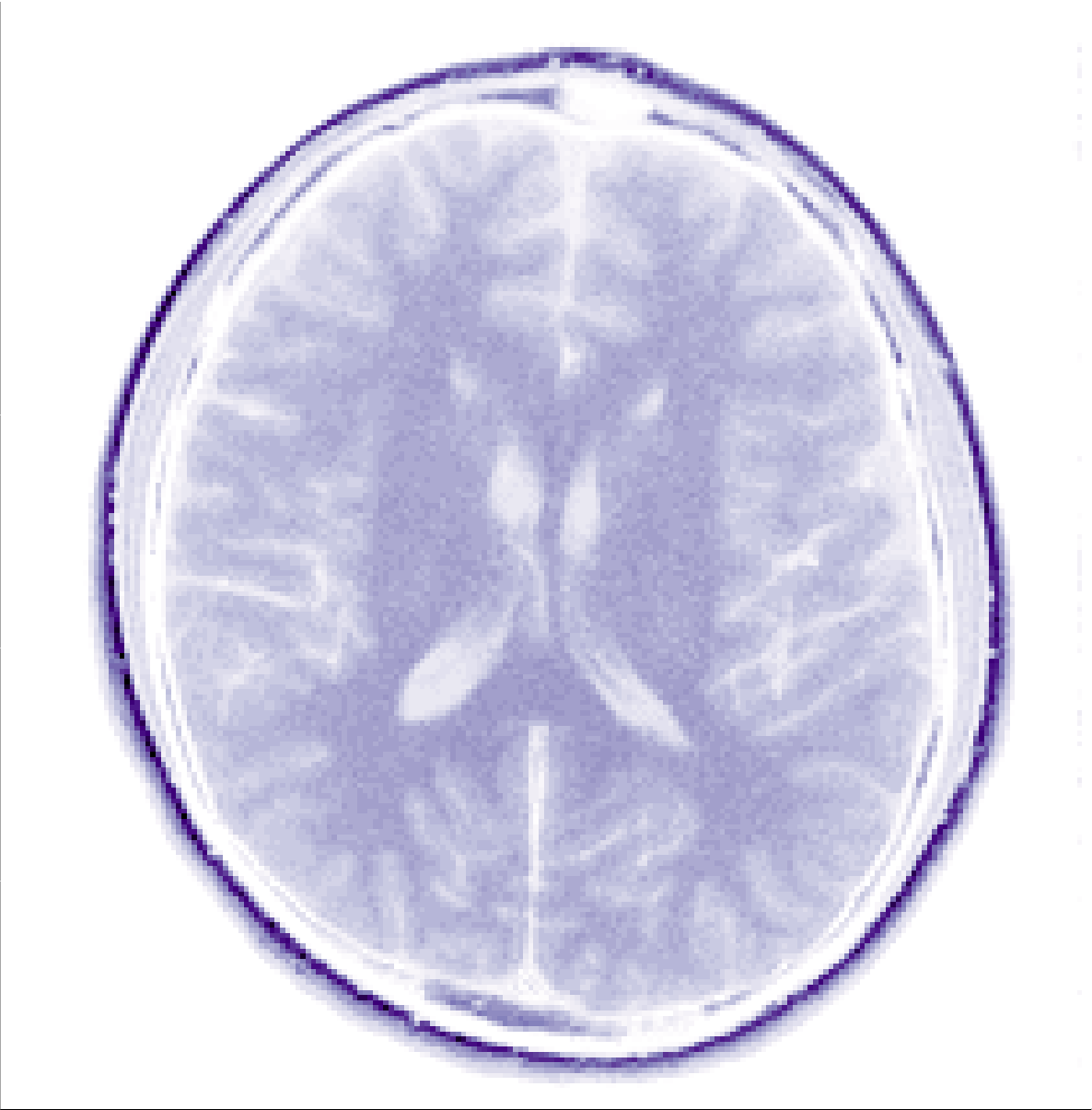}
		 } 
		 \subfloat{
			\includegraphics[width=0.3\linewidth, page=2]{figures/pred-crop.pdf}
		} 
		\subfloat{
			\includegraphics[width=0.3\linewidth, page=3]{figures/pred-crop.pdf}
                      }                    
 }
                    \mbox{
	   \subfloat{
		 \includegraphics[width=0.3\linewidth, page=1]{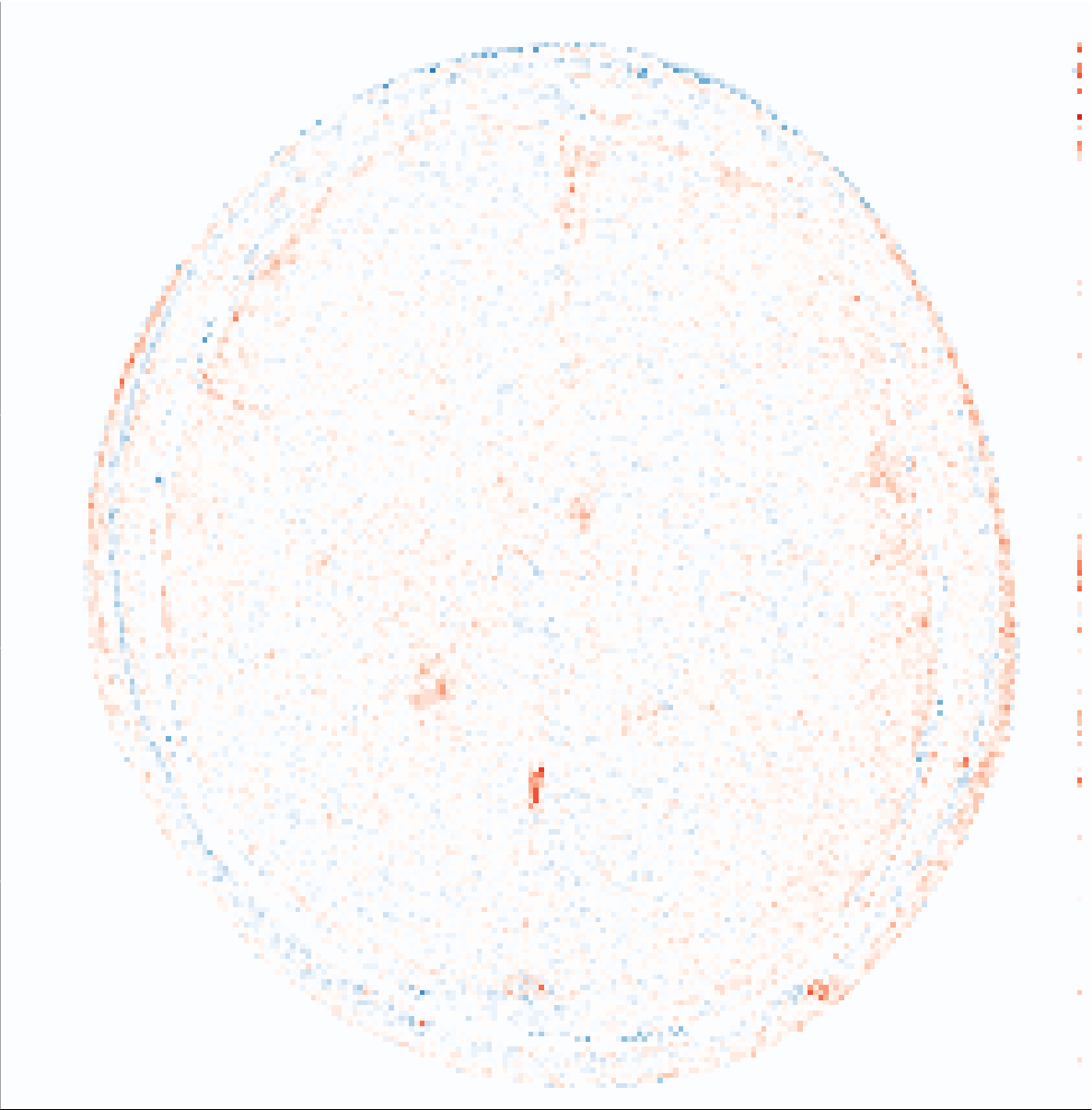}
	   } 
	   \subfloat{
		\includegraphics[width=0.3\linewidth, page=1]{figures/diff-crop.pdf}
	  } 
		\subfloat{
		\includegraphics[width=0.3\linewidth, page=3]{figures/diff-crop.pdf}
              }
              }
	   \caption{(Top) Mid-axial views of the OSL, DL and
             AECM-synthetized images for $\TE/\TR = 15/0.6$, along with
             (bottom) the differenced image (from the acquired image)
             at that setting.
		   }
	   \label{fig:3dplots_all_mu}
   \end{figure}
Figure~\ref{fig:3dplots_all_mu} displays the performance of OSL, DL, PLS and
AECM at a sample setting. A closer inspection of the  OSL and DL
synthetic MR images indicate greater disagreement from the truth, than
with the AECM and PLS image.
Further, following~\citet{maitraandriddles10}, we also considered performance
of the AECM methodology in synthesizing MR images  
as the number of images in the training set is increased. 
Figure~\ref{fig:3D_consistency} in Section~\ref{sec:supp-consistency} 
shows improvement of the performances 
over the best set of images, indicating consistency of our methodology.

\subsubsection{Estimating SEs of ROI means}
We also evaluated our estimated  SEs of ROI means of
synthetic images obtained using the training set of the real data. 
We considered three ROIs consisting of 
Cerebro-Spinal Fluid (CSF), Gray Matter (GM) and White Matter
(WM). Table~\ref{table:variance} provides the ROI means of the
synthetically generated variances. The machinery of
Section~\ref{variance} was used to estimate SEs of the ROI means
and compared with computationally expensive boostrap
estimates. (For details, see Section~\ref{sec:supp-para-boot}, which
also shows insignificant differences between our estimates and the 
bootstrap-estimated SEs.)
 
Note also that in the AECM case, the bootstrap-estimated SEs are
computationally expensive to obtain, but not prohibitive as in the
case of SEs of DL-estimated regional means because the DL-estimation
procedure would take almost a day for each bootstrap replication. In
contrast, our AECM estimate takes a few minutes.
\begin{table}[h]
	\centering
	\caption{Estimated CSF, GM and WM ROI means in synthetic
          images and their SEs: $\widehat\sigma_s$s calculated using
          the formula in Section~\ref{variance}, and
          $\widehat\sigma_b$s estimated using parametric
          bootstrap. The scale for the SEs is $\times10^{-3}$.} 
	\label{table:variance}
	    \begin{tabular}{c|rcc|rcc|rcc}
			    \hline 
			    \text{Image}  & \multicolumn{3}{c|}{\text{CSF}} & \multicolumn{3}{c|}{\text{GM}} & 
			    \multicolumn{3}{c}{\text{WM}} 
				\\
				\cmidrule{2-10}
				\text{$j$ } &  Mean & $\widehat\sigma_s$ & $\widehat\sigma_b$ &  Mean & $\widehat\sigma_s$ & $\widehat\sigma_b$ 
				&  Mean & $\widehat\sigma_s$ & $\widehat\sigma_b$\\
				\hline
				2 & 36.7 & 3.73 & 4.45 & 73.1 & 3.31 & 3.71 & 123.4 & 7.53 & 6.93\\ 
				3 & 33.8 & 3.95 & 4.15 & 67.9 & 3.05 & 3.41 & 111.4 & 7.07 & 6.61 \\ 
                4 & 57.7 & 5.11 & 5.62 & 112.7 & 4.09 & 4.62 & 167.1 & 7.73 & 7.38 \\ 
                5 & 42.0 & 4.13 & 4.35 & 84.0 & 2.99 & 3.28 & 110.5 & 6.22 & 6.07 \\ 
                6 & 36.3 & 3.85 & 4.07 & 72.8 & 2.74 & 2.95 & 90.0 & 6.44 & 6.52 \\ 
                7 & 88.4 & 5.55 & 5.70 & 164.0 & 3.96 & 4.23 & 191.4 & 9.86 & 9.52 \\ 
                8 & 55.9 & 4.44 & 4.51 & 106.4 & 2.79 & 2.77 & 102.7 & 7.21 & 7.24 \\ 
                11 & 52.6 & 5.42 & 5.41 & 93.4 & 3.59 & 3.73 & 69.8 & 7.44 & 7.80 \\ 
                12 & 33.3 & 5.13 & 4.69 & 54.5 & 3.39 & 3.49 & 31.1 & 5.74 & 6.18 \\ 
			\hline
		\end{tabular}
\end{table}

\subsubsection{Choice of Optimal Training Set}


Following~\citet{maitraandriddles10}, we investigate all possible 
$(\TE, \TR)$ triplets to see the best choice for synthesizing images.
For each possible set of three training images, we obtained LS, OSL-EM and
AECM-estimated parameters and then obtained synthetic images for the
rest images considered test set. Section~\ref{sec:supp-optimal-choice} 
shows that the best synthetic images are indeed obtained when we have  $\rho-$,
$\T_1$-, and  $\T_2$-weighted images as the triplet of training images.

\subsection{Simulation Experiments}
Our next set of experiments is on prediction accuracy
relative to noise. We demonstrate and evaluate performance of our 
methodology on images obtained at a 1mm$\times$1mm$\times$5mm
resolution with 
181$\times$217$\times$36 voxels and simulated using the Brainweb 
interface~\citep{cocoscoetal97}.
This tool simulates realistic images at TE, TR, and flip angle settings, slice thicknesses, noise levels, and intensity nonuniformity (INU) proportions.
We collected training images from Brainweb's ``mild" and ``severe" multiple sclerosis lesions databases and 
used them in a simulated spin-echo MR sequence with the same three $(\TE, \TR) $
parameters and noise levels $\{1\%,2.5\%,5\%,7.5\%,10\%\}$. 
Separately, we acquired spin-echo training images with 1\% noise for the same (TE,TR) parameters to investigate the influence of field inhomogeneity 
 at $\{0\%, 5\%, 10\%\}$ INU (of so-called field A) levels.
We acquired noiseless ``ground truth" images from the Brainweb interface at 0\% INU levels 
and nine $(\TE, \TR)$s to evaluate our predicted synthetic images.

We compared the predicted synthetic images at the nine (unshaded) 
$(\TE,\TR)$ values of RMSPE of the
predicted synthetic MRIs for the ``mild'' MS dataset in Figure~\ref{fig:BW_RMSPE}. 
We see that the AECM provides better
prediction accuracy than OSL and at a given setting and also beats PLS. 
At lower noise, AECM outperforms DL, while DL outperforms AECM for higher noise.
The performances for the ``severe'' MS dataset and other scenarios are described in Section~\ref{sec:supp-brainweb-others}. 
\begin{figure}[h]
	\centering
    \includegraphics[width=\linewidth, page=10]{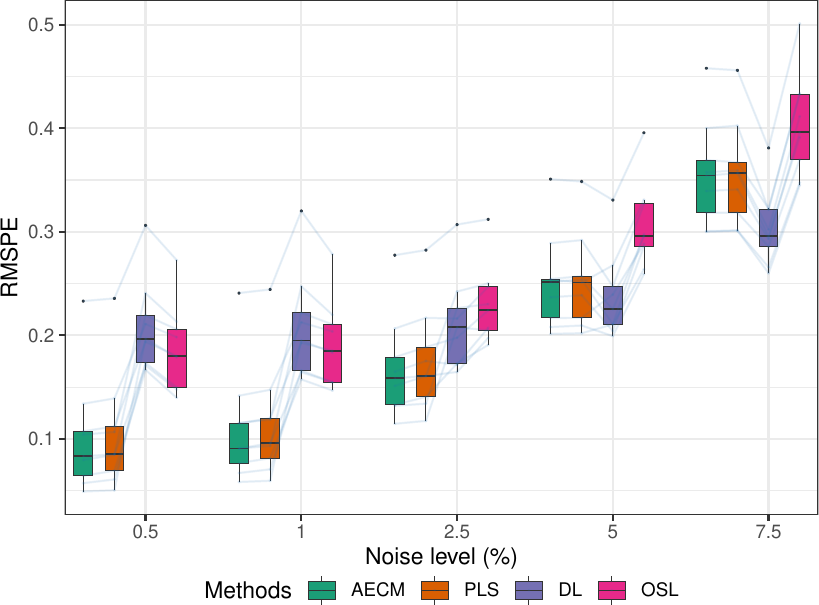}  
	\caption{RMSPE values for different noise of Brainweb
          simulated data for different methods for ``mild'' MS
          dataset. The lines in the linked boxplot connect the
          different datasets.}
	\label{fig:BW_RMSPE}
\end{figure}
Our performance evaluations indicate our AECM approach can
improve synthetic MRI both in simulation and real data settings. 
Further, AECM comparable to LS method in terms of time  and better in terms of performances.
It is far quicker than DL (with AECM taking about a  minute, and the latter taking almost a day), synthesizing
images  in real time and yielding superior results in realistic
low-noise situations.


\section{Discussion}
\label{sec:dis}
We have provided a computationally practical but theoretically sound
implementation of model-based synthetic MR imaging that can enable its
use patient-specific settings. Our approach penalizes the
Rice-distribution based loglikelihood through a transformed GMRF~\citep{maitraandriddles10} with parameters estimated via a
matrix-free AECM scheme, that unlike OSL-EM, is guaranteed to
solve~\eqref{eq:penalizedproblem} locally. 
Over different simulations, the block-descent-optimization-based PLS 
gives  comparable results, but is edged out by AECM. Further, while
PLS estimates are obtained at about half the time as AECM estimates,
the latter itself takes less than a couple of minutes, even on a large
image volume, and is practical to apply in a clinical setting,
and can  be obtained and
interpreted during a patient visit. With the same resources, the
DL-based method took about a day. 
Experimental evaluations also show AECM to have better overall performance over even DL-based synthetic MR,
at substantial lower cost. An added benefit of our model-based
approach is the ready availability of SE estimates, for which we
also provide practical matrix-free computations. Our experiments show
our estimates to be nonsignificantly distinct from those
bootstrap-obtained values. Compute costs make
parametric bootstrap-estimated SEs impractical to implement for
the DL-based method where  (approximate) analytical expressions for
the regional contrast SEs are also not readily obtained.

A reviewer has asked about the context of our results in light of 
\citet{adrianetal13}'s findings that Ricean-based activation tests in functional Magnetic Resonance Imaging (fMRI) analysis are only superior to
        Gaussian-approximation-based tests for small signal-to-noise
        ratios (SNRs) of less than 0.6, typically much smaller than
        what is seen in 
        practice. We note that those findings are on activation
        detection and under the assumption of independence in the fMRI
        time series. Indeed, \citet{adrianetal18} demonstrated that using the
        original Gaussian complex-valued data (which forms the basis
        of the Ricean model for the magnitudes) and an AR($p$)
        temporal dependence 
        and a general covariance structure for the real and imaginary
        errors, showed greater power of activation than using a
        Gaussian AR($p$) model on the magnitude data.
(Subsequent recent work in \citet{adrianetal23} supports these
findings also for the Rice-distributed time series data.)
        Finally, we note
        that our methods here can benefit MRI data acquired at higher
        resolution, since voxel volume is inversely proportional to
        SNR \citep{adrianetal18}. Therefore, we believe that the use
        of the more accurate Rice model \eqref{eq:rice} over its
        Gaussian approximation has value, especially in the
        context of higher-resolution synthetic MRI studies. 

Both reviewers have asked about the feasibility of our method  with
penalties more general than the first order MRF in \eqref{matrixnormal}.
For such  neighborhood structures ({\em e.g.}, higher order MRFs) on $\W$, the partitioning of $\W$ needs to be done through coding sets \citep{besag74} or concliques \citep{kaiseretal12} that generalize the checkerboard bifurcation. Succinctly, for higher order MRFs, voxels can be partitioned into disjoint sets ${\cal C}_1,{\cal C}_2,\ldots, {\cal C}_l$ for some $l \in \mathbb{N}$ so that for each $j$, $\{\W_i ~:~ i\in{\cal C}_j\}$ are conditionally independent random variables given the rest. Consequently, the optimization problems on each ${\cal C}_j$  $(1\le j \le l)$ remain embarrassingly parallel. Further, the eigenvalues of the matrix $\bLambda$ are analytically known for many of these higher order MRFs \citep{mondal18} on 2D lattice, and we surmise these can be also computed in 3D.

There are some other issues that could benefit from further
attention. For instance, transformations
beyond the ones used in transforming $(\rho,\T_1,\T_2)$ to
$(\W_1,\W_2,\W_3)$ could be investigated. Also of interest may be the choice of penalty functions that go beyond stationary GMRFs. The use of 
practical Bayesian methodology incorporating ingredients of our
matrix-free approach may also be worth exploring. 

\section*{Acknowledgments}
A portion of this article
won the first author a 2023 Student Paper Competition award from the
American Statistical Association (ASA) Section on Statistical Computing
and Graphics. 
The research of the first and third authors was supported in part by
the National Institute of Biomedical Imaging and Bioengineering (NIBIB) 
of the National Institutes of Health (NIH) under Grant R21EB034184.
The content of this paper is however solely the responsibility of the
authors and does not represent the official views of the  NIBIB or the NIH.

\bibliographystyle{IEEEtran}
\bibliography{variance}

\setcounter{equation}{0}
\setcounter{figure}{0}
\setcounter{table}{0}
\setcounter{section}{0}

\renewcommand\thefigure{A\arabic{figure}}
\renewcommand\thetable{A\arabic{table}}
\renewcommand\thesection{A}
\renewcommand\thesubsection{\Alph{section}-\Alph{subsection}}
\renewcommand\theequation{A\arabic{equation}}

\section{Appendix}

\section{Proofs}
\subsection{Proof of Proposition~\ref{prop:optimization_W_i}}
\label{proof1}
Only the kernel of $f(\bW;\bPsi,\bLambda)$ in~\eqref{matrixnormal} involves $\bW$ so we
have, but for a constant free of $\bW$, 
\begin{equation} \label{eq:MRF}
		\log f(\bW;\bPsi,\bLambda) = -\frac12 \Tr(\bPsi^{-1} \W' \bLambda \W) =
                -\frac12 \Tr\left(\bLambda (\W \bPsi^{-1} \W')\right)
                = -\frac12 \sum_p \sum_q \Lambda_{p,q} \cdot \W_q \bPsi^{-1} \W_p'.
\end{equation}
The portion of~\eqref{eq:MRF} that involves the $i$th voxel, and hence
$\bW_i$, is 
\begin{equation}\label{eq:MRF_part}
	\log f(\bW;\bPsi,\bLambda) = c -\Big(\sum_{q\ne i} \Lambda_{i, q} \W_q\Big) \bPsi^{-1} \W_i' 		\  -\frac12 \Lambda_{i, i} \W_i' \bPsi^{-1} \W_i
\end{equation}
which gives us equation~\eqref{eq:Q_star_i}, with $c$ a constant not depending on $\W_i$. However, $\bLambda\neq\bzero$ only if $q \heartsuit i$. 
Applying the chain rule for derivatives, we get the first partial
derivative, w.r.t. $W_{ik}$, of the $Q(\bW;\bW^{(t)})$ of (\ref{eq:Q_fn}) as
\begin{equation}
	\frac{\partial Q}{\partial W_{ik}} 
	= \sum\limits_{j=1}^m \sigma_j^{-2} \left(
	- \nu_{ij} + r_{ij} Z^{(t)}_{ij} \right)
	\frac{\partial \nu_{ij}}{\partial W_{ik}},
        \label{eq:grad}
      \end{equation}
      for which we need the derivatives of~\eqref{eq:bloch}.  
For the transformed variable $\bW$, the Bloch equation is
\begin{equation}
    \nu_{ij} = W_{i1} \left\{1-W_{i2}^{\TR_j}\right\} W_{i3}^{\TE_j},
    \label{eq:bloch_2}
\end{equation}
where from we get the first partial derivatives
\begin{equation*}
\frac{\partial \nu_{ij}}{\partial W_{ik}} = \begin{cases}
W_{i3}^{\TE_j}\left\{1-W_{i2}^{\TR_j}\right\}, &\qquad\!\!\!\!\!k\!=\!1\\
-W_{i1}  \TR_j  W_{i3}^{\TE_j}W_{i2}^{\TR_j-1}, &\qquad\!\!\!\!\!k\!=\!2\\
W_{i1}  \TE_j   W_{i3}^{\TE_j-1}\left\{1-W_{i2}^{\TR_j}\right\}, &\qquad\!\!\!\!\!k\!=\!3\\
\end{cases}
\end{equation*}
and the second partial derivatives
\begin{align*}
\frac{\partial^2 \nu_{ij}}{\partial W_{ik}\partial W_{ik'}} = 
& \begin{cases}
0, &\qquad\!\!\!\!\!k\!=\!1,\!k'\!=\!1\\
- \TR_j  W_{i3}^{\TE_j}W_{i2}^{\TR_j-1}, &\qquad\!\!\!\!\!k\!=\!1,k'\!=\!2\\
\TE_j W_{i3}^{\TE_j-1}(1-W_{i2}^{\TR_j}), &\qquad\!\!\!\!\!k\!=\!1,k'\!=\!3\\
- W_{i1} \TR_j(\TR_j-1) W_{i3}^{\TE_j}W_{i2}^{\TR_j-2}, &\qquad\!\!\!\!\!k\!=\!2,k'\!=\!2\\
- W_{i1} \TR_j  \TE_j  W_{i3}^{\TE_j-1} W_{i2}^{\TR_j-1}, &\qquad\!\!\!\!\!k\!=\!2,k'\!=\!3\\
W_{i1} \TE_j(\TE_j-1)W_{i3}^{\TE_j-2}(1-W_{i2}^{\TR_j}), &\qquad\!\!\!\!\!k\!=\!3,k'\!=\!3.\\
\end{cases}
\end{align*}
From the above, we get \eqref{eq:grad_Q_star_i} from~\eqref{eq:grad}
and~\eqref{eq:MRF_part}.
\subsection{Proof of Proposition~\ref{prop:profileout}}
\label{proof2}
From \eqref{matrixnormal}, we have
\begin{equation}
  \begin{split}
\log f(\bW;\bPsi,\bbeta) & =
                           -\frac{1}{2}\Tr\left(\bPsi^{-1}\W'\bLambda
                           \W\right) + \frac{3}{2}\log |\bLambda|^{*}
                           - \frac{3n}{2} \log (2\pi) - \frac{n}{2}
                           \log |\bPsi|  \\
    &= -\frac{1}{2}\Tr\left(\bPsi^{-1} \bm\Xi\right)  - {\frac{n}{2}} \log |\bPsi| + K, \end{split} \label{eq:Psi_beta}
\end{equation}
where $\bm\Xi = (\W'\bLambda \W)$ and $K$ depends only on $\bbeta$ and
$n$. 

The optimization of \eqref{eq:Psi_beta} w.r.t. $\bPsi$ is
similar to the  problem for finding the maximum likelihood estimator
for the variance-covariance matrix in multivariate normal samples as
in \citet[Theorem 4.2.1]{mardiaetal79}. In our case, the first two
terms contain the arithmetic (AM) and the geometric means (GM) of
eigenvalues of $n^{-1}\bPsi^{-1}\bm\Xi$. The AM-GM inequality shows
that  this expression is maximized at $\displaystyle \bPsi =
{\bm\Xi}/{n} = {(\W'\bLambda \W)}/{n}$. Incorporating this estimated
value yields the profile likelihood
\begin{equation}\label{eq:Psi_beta_2}
          Q_p^{*}(\bbeta) = -\frac{1}{2}\Tr\left(n \bI_3\right) + \frac{3}{2}\log |\bLambda|^{*} - \frac{3n}{2} \log (2\pi) - \frac{n}{2} \log \left|\frac1n \bm\Xi\right| 
          = c+ \frac{3}{2}\log |\bLambda|^{*} - \frac{n}{2} \log |\W'\bLambda \W| 
\end{equation}
where $c$ is a constant that involves $n$.

\section*{Supplementary Information}
\setcounter{equation}{0}
\setcounter{figure}{0}
\setcounter{table}{0}
\setcounter{section}{0}

\renewcommand\thefigure{S\arabic{figure}}
\renewcommand\thetable{S\arabic{table}}
\renewcommand\thesection{S\Roman{section}}
\renewcommand\thesubsection{S\Roman{section}-\Alph{subsection}}
\renewcommand\theequation{S\arabic{equation}}

\section{Supplementary materials for Theory and Methods}\label{sec:sim-meth}
\subsection{Information Matrix}
\label{sec:supp-info-mat}
The observed information matrix corresponding to $\W$ is obtained
from~\eqref{llhd} and~\eqref{eq:rice}. The first partial derivative is similar to~\eqref{eq:Q_star_i}: 
\begin{equation*}
\frac{\partial \ell}{\partial W_{ik}} 
= \sum\limits_{j=1}^m \sigma_j^{-2} \left\{ 
- \nu_{ij} +
r_{ij} \frac{\I_1\left(\frac{r_{ij}\nu_{ij}}{\sigma_j^2}\right)}{\I_0\left(\frac{r_{ij}\nu_{ij}}{\sigma_j^2}\right)}\right\}
\frac{\partial \nu_{ij}}{\partial W_{ik}},
\end{equation*}
while the second partial derivative is
\begin{equation*}
 \frac{\partial^2 \ell}{\partial W_{ik'} \partial W_{ik}}
 = \sum\limits_{j=1}^m  \sigma_j^{-2} \left\{
- \nu_{ij} +
r_{ij} \frac{\I_1\left(\frac{r_{ij}\nu_{ij}}{\sigma_j^2}\right)}{ \I_0\left(\frac{r_{ij}\nu_{ij}}{\sigma_j^2}\right)}\right\}
\frac{\partial^2 \nu_{ij}}{\partial W_{ik'} \partial W_{ik}}
+ \sum\limits_{j=1}^m \Bigg\{ -\frac{1}{\sigma_j^2} + \frac{r_{ij}^2}{\sigma_j^4}  h\left(\frac{r_{ij}\nu_{ij}}{\sigma_j^2}\right)
\Bigg\}  \frac{\partial \nu_{ij}}{\partial W_{ik}}   \frac{\partial \nu_{ij}}{\partial W_{ik'}}.
\end{equation*}
$\widehat\bH$ is equal to the negative second derivative w.r.t. $W_{ik}$ over $i$ and
$k$, and $\I_s(x)$ is the modified Bessel function of the first kind
of order $s$.  Note that, any order of derivative of $\ell_i$ w.r.t. $W_{i'k}$ is $0$ if $i\ne i'$, making $\widehat\bH$ block diagonal, and 
\begin{equation}
    h(x) := \frac{d}{dx} \left(A_1(x)\right) = \frac12 \left[\frac{\I_0(x)\left\{\I_0(x)+\I_2(x)\right\} - 2\I_1^2(x)}{\I_0^2(x)}\right]
\end{equation}
The part of the penalized likelihood from the Rice density contributes to $\widehat\bH$, and the Hessian of the penalty part contributes to  $\bLambda\otimes\bPsi^{-1}$,  producing the $\widehat\bOmega$ matrix. 

\section{Supplementary Materials for Performance Evaluations}
\subsection{The trajectory of the objective function
in~\eqref{eq:penalizedproblem} when using OSL-EM and AECM}
\label{sec:supp-trajectory-likeli}
\begin{figure}
  \vspace{-0.2in}
  \centering
	\includegraphics[width=\linewidth]{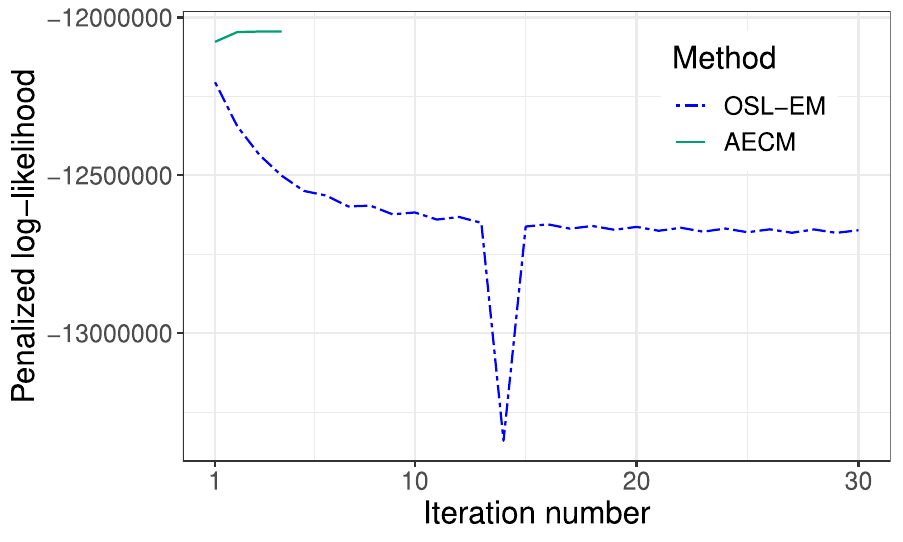}
	\caption{The path of the objective function
          in~\eqref{eq:penalizedproblem}, using OSL-EM and AECM.} 
	\label{fig:likeliplot}
      \end{figure}
Figure~\ref{fig:likeliplot} shows the path of the objective function
in~\eqref{eq:penalizedproblem} for the example in
Section~\ref{sec:ill}. We see that the penalized log-likelihood for
OSL-EM does not always increase in every iteration. In fact, in this
example, it oscillates, especially with increased desired precision in
terms of relative error. However, the penalized log-likelihood
increases for AECM, as is guaranteed by theory. Further, with the
increasing ascent property, the algorithm usually converges with fewer
iterations than OSL-EM. 
\subsection{Bootstrap validation of SE calculations}
\label{sec:supp-para-boot}
We validate our theoretically derived approximate SEs using the
parametric bootstrap. Our bootstrap procedure was as follows. First,
using the estimated $\widehat{\W}$, we obtained resampled training
images were generated using the equations~\eqref{eq:bloch}
and~\eqref{eq:rice} with the train \TE \ and \TR \ settings. For the
$b$th resampled training images, we used AECM to have estimated
$\widehat{\W}^*_b$, $b=1,2,\ldots,B$, with $B$ as the number of
bootstrap replications, and inserted them, along with the test
$(\TE,\TR)$~\eqref{eq:bloch} used to get $\widehat{\bnu}_b$ for
these settings. From these resampled synthetic images, the standad
deviations of ${\bc}'\widehat{\bnu}_b$, over $b=1,2,\ldots,B$,  can be
calculated to provide us with estimated SEs of ${\bc}'\widehat{\bnu}$
which can be compared to our theoretically derived approximate SEs.
However, we note that our validator has sampling variability. In order
to account for this variability, we calculated the standard deviations
of these bootstrap SEs. We used a jackknife-type
estimator to estimate this SD. Specifically, from the bootstrap
replications, we used a leave-one-out method to get
$B$ leave-one-out bootstrapped SE estimates. The standard deviation
over these $B$ leave-one-out SE estimates can be used to obtain a SE
of the boostrap-estimated SE. Our experiments indicate that the
theoretical SEs of the ROI means of the synthetic images are within
one SE of the bootstrap-estimated SEs and so are not significantly
different.

\subsection{Choice of Optimal Training Set}
\label{sec:supp-optimal-choice}
Following~\citet{maitraandriddles10}, we also investigate all possible
$(\TE, \TR)$ settings to see if some choice can 
provide better synthetic images. There are ${}^{12}C_{3}=220$
possible sets of three training images, but some of them are not
distinct and are discarded, leaving behind 212 sets.
For each set of three training images, we obtained LS, OSL-EM and
AECM-estimated parameters and then obtained synthetic images for the
nine $(\TE, \TR)$-values outside the considered training set. 

Table~\ref{table:3D_best_10_whole_v2} displays numerical performance of 
the AECM-estimated synthetic images for the top 10 combinations 
(ordered according to the increasing scaled RMSPE), and shows the 
per cent improvement of AECM and OSL-EM over LS for each of the measures. 
We see around 5\% improvement of AECM over LS. Once again, OSL-EM is 
unpredictable in its improvement over LS, often doing worse.


\begin{table*}[h]
  \centering
  \caption{\label{table:3D_best_10_whole_v2}The 12 $(\TE , \TR)$ settings (in seconds) for the data, and the ten best training sets, their performance
    measures  for the LS estimates, and their average
    performance measures relative to the LS estimates
    ($\times 0.01$). In the table, OSL and AE denote OSL-EM and
    AECM-estimated predictions.}
     \begin{tabular}{lc|rrrrrrrrrrrr}
			  \hline
			  \multirow{3}{*}{Settings} & j & 1 & 2 & 3 & 4 & 5 & 6 & 7 & 8 & 9 & 10 & 11 & 12 \\
			  \hline
			  & TE  & 0.01 & 0.015 & 0.02 & 0.01 & 0.03 & 0.04 & 0.01 & 0.04 & 0.08 & 0.01 & 0.06 & 0.1 \\
			  & TR  & 0.6 & 0.6 & 0.6 & 1.0 & 1.0 & 1.0 & 2.0 & 2.0 & 2.0 & 3.0 & 3.0 & 3.0 \\
			  \hline
		  \end{tabular}
		  \vspace{0.5cm}
		  
  \begin{tabular}{crrrrrrrrrrrr}
    \hline
    & \multicolumn{4}{c}{LS performance measure} & \multicolumn{4}{c}{RMSPE} & \multicolumn{4}{c}{MAPE}\\
    \cmidrule{2-13}
    Training & \multicolumn{2}{c}{RMSPE} & \multicolumn{2}{c}{MAPE} & \multicolumn{2}{c}{$\widehat\nu$} & \multicolumn{2}{c}{$\widehat\nu^*$}  & \multicolumn{2}{c}{$\widehat\nu$} & \multicolumn{2}{c}{$\widehat\nu^*$}
    \\
    \cmidrule{2-13}
    Images & $\widehat\nu$ & $\widehat\nu^*$ & $\widehat\nu$ & $\widehat\nu^*$  & OSL & AE & OSL & AE & OSL & AE & OSL & AE \\
    \hline
    1, 9, 10 & 20.93 & 20.92 & 24.89 & 24.89 & -0.44 & 1.10 & -0.34 & 1.98 & -0.26 & 1.28 & 0.01 & 1.20  \\ 
    2, 9, 10  & 21.11 & 21.10 & 25.08 & 25.07 &-0.64 & 1.28 & -0.50 & 2.41 & -0.47 & 1.49 & -0.09 & 1.44 \\ 
    2, 10, 12 & 21.60 & 21.60 & 25.71 & 25.70 & -1.19 & 0.95 & -1.54 & 1.76 & -1.21 & 1.00 & -0.80 & 0.95  \\ 
    2, 7, 9 & 21.38 & 21.35 & 25.34 & 25.32 & 0.68 & 3.42 & 1.53 & 6.55 & 1.23 & 3.51 & 1.32 & 3.44 \\ 
    1, 10, 12 & 21.64 & 21.64 & 25.77 & 25.77 & -0.90 & 1.12 & -1.26 & 1.92 & -1.09 & 1.15 & -0.81 & 1.03  \\ 
    1, 7, 9 & 21.54 & 21.54 & 25.54 & 25.54 & 0.14 & 2.97 & 0.77 & 5.78 & 0.65 & 2.96 & 0.87 & 2.94 \\ 
    1, 8, 9 & 21.54 & 21.53 & 24.97 & 24.96 & 0.19 & 3.82 & 0.85 & 7.36 & -0.12 & 3.46 & 0.10 & 3.41 \\ 
    1, 10, 11 & 21.86 & 21.83 & 25.83 & 25.81 & 0.02 & 1.05 & 0.57 & 2.01 & 0.15 & 1.08 & 0.36 & 1.05 \\ 
    2, 10, 11 & 21.91 & 21.89 & 25.89 & 25.86 & -0.21 & 1.27 & 0.26 & 2.44 & -0.13 & 1.43 & 0.19 & 1.40 \\
    2, 8, 9 &  21.60 & 21.59 & 25.05 & 25.04 &-0.00 & 3.24 & 0.59 & 6.35 & -0.38 & 3.03 & -0.09 & 3.02  \\ 
    \hline
  \end{tabular}
\end{table*}

\subsection{Consistency of Statistical Methods}
\label{sec:supp-consistency}
\begin{figure}
  \vspace{-0.2in}
  \centering
  \includegraphics[width=\linewidth]{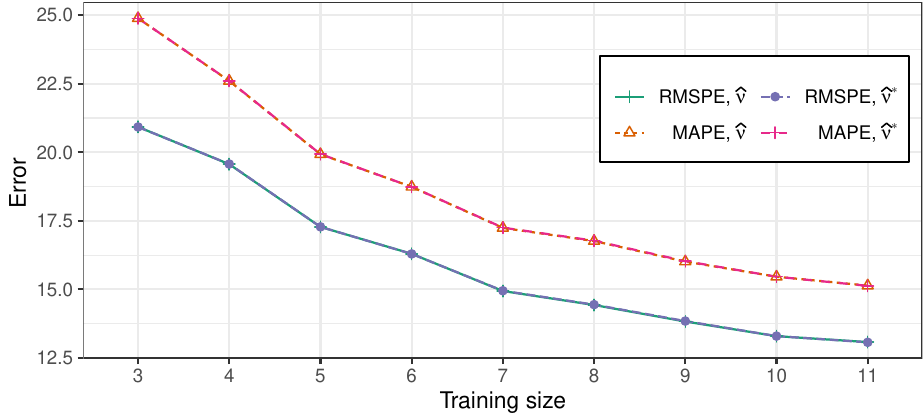}
	\caption{The best training samples for varying training sample
          sizes and their performance measures.}
	\label{fig:3D_consistency}
\end{figure}
A desirable feature of any estimation or prediction method is
statistical consistency, or improvement in performance as more
information (training image set) becomes available. Therefore, we
considered and evaluated performance of the AECM method in generating
synthetic images as the number of images in the training set increased
over $m\in\{3, 4, 5,\dots, 11\}$. For each $m$, we evaluated
predictive performance with a training set of images of all possible training image set combinations of size $m$ and compared them with the remaining images in the test sample. The RMSPE and MAPE
values for the best set, and for each $m$ are displayed in
Figure~\ref{fig:3D_consistency} and show consistency of our synthetic
image generation method.

\subsection{Performance measures for the Brainweb data under different conditions }
\label{sec:supp-brainweb-others}
Figure~\ref{fig:BW_RMSPE_severe} shows a similar set of performance measures as 
in~\ref{fig:BW_RMSPE}, but with ``severe'' MS dataset with various noise percentages. 
\begin{figure}[!t]
	\centering
    \includegraphics[width=0.7\linewidth, page=10]{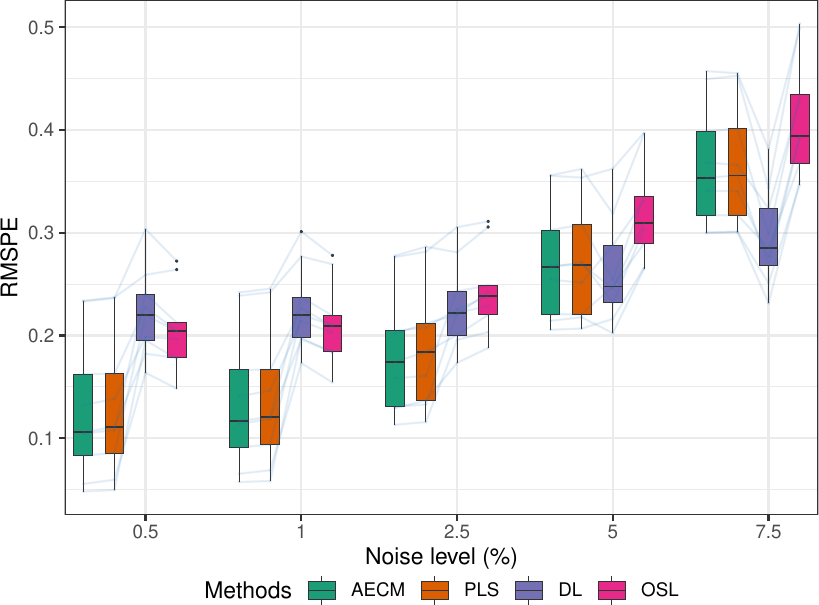}  
	\caption{RMSPE values for different noise of brainweb simulated pictures for different methods for ``severe'' MS dataset.}
	\label{fig:BW_RMSPE_severe}
\end{figure}

We have also evaluated our methods for various INU level, which is presented in the Figure~\ref{fig:BW_RMSPE_INU}. 
\begin{figure}[!t]
	\centering
    \includegraphics[width=0.7\linewidth, page=2]{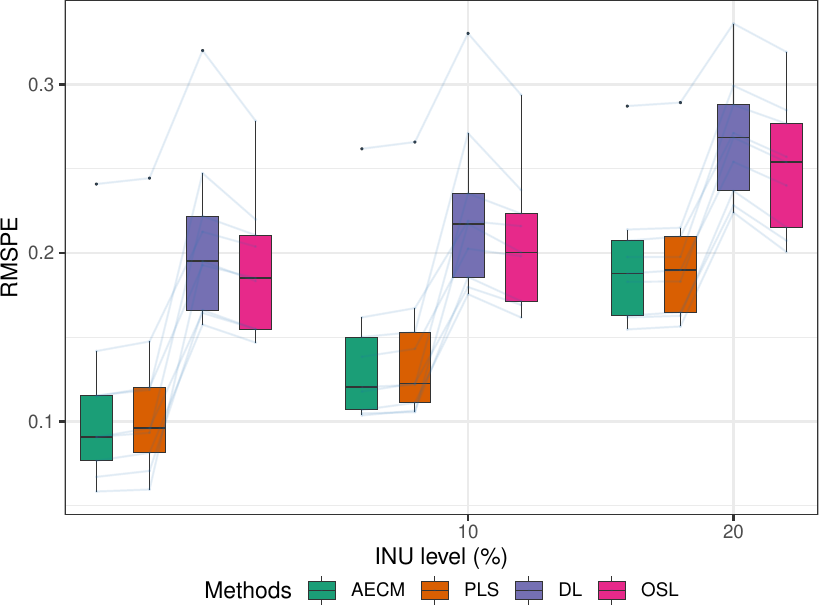}
	\caption{RMSPE values for different noise of brainweb simulated pictures with respect to different INU levels.}
	\label{fig:BW_RMSPE_INU}
\end{figure}


\end{document}